\documentclass[amsa]{ipart}
\usepackage{epsfig}
\usepackage{graphicx}
\RequirePackage[OT1]{fontenc}
\RequirePackage{amsthm,amsmath,amssymb}
\RequirePackage[numbers,square]{natbib} %
\RequirePackage[colorlinks]{hyperref}
\usepackage{amsopn}
\usepackage{mathrsfs}
\newtheorem{thm}{Theorem}
\newtheorem{exmp}{Example}[section]
\newtheorem{prop}{Proposition}
\newtheorem{rmk}{Remark}

\newcommand {\pnull}{P(H_0|\text{Data})}
\newcommand {\palt}{P(H_1|\text{Data})}
\newcommand {\bfactor}{\text{Bayes factor}}
\newcommand {\postnull}{P(\text{Data}|H_0)}
\newcommand {\postalt}{P(\text{Data}|H_1)}

\newcommand{\ee}{\mathbb{E}}

\newcommand{\rr}{\mathbb{R}}
\newcommand{\R}{\mathbb{R}}
\pubyear{}
\volume{}
\issue{}
\firstpage{}
\lastpage{}

\begin{document}

\begin{frontmatter}
\title{Bayesian Goodness of Fit Tests: A Conversation for David Mumford}
\runtitle{Bayes Test Goodness of Fit}

\begin{aug}
\author{\fnms{Persi} \snm{Diaconis}\thanksref{t1}\ead[label=e1]{diaconis@math.stanford.edu}},
\address{Department of Mathematics and Statistics\\
Stanford University\\
\printead{e1}}
\author{\fnms{Guanyang} \snm{Wang}\ead[label=e2]{guanyang@stanford.edu}}
\address{Department of Mathematics\\
Stanford University\\
\printead{e2}}

\thankstext{t1}{Research partially supported by NSF Grant DMS-1208775.}

\runauthor{Persi Diaconis and Guanyang Wang}

\end{aug}

\begin{abstract}
The problem of making practical, useful goodness of fit tests in the Bayesian paradigm is largely open. We introduce a class of special cases (testing for uniformity: have the cards been shuffled enough; does my random generator work) and a class of sensible Bayes tests inspired by Mumford, Wu and Zhu \cite{Zhu:1998:FRF:290091.290092}. Calculating these tests presents the challenge of `doubly intractable distributions'. In present circumstances, modern MCMC techniques are up to the challenge. But many other problems remain. Our paper is didactic, we hope to induce the reader to help take it further.
\end{abstract}


\end{frontmatter}

\section{Introduction}

Here is one of David Mumford's surprising accomplishments: he's managed to make the transition from abstract to applied mathematics, along the way becoming a card-carrying statistician. He tells his own version of the story in \cite{mumford2000dawning}. ``\,I also have to confess at the outset to the zeal of a convert, a born again believer in stochastic methods. Last week, Dave Wright reminded me of the advice I had given a graduate student during my algebraic geometry days in the 70's :``\,Good Grief, don't waste your time studying statistics. It's all cookbook nonsense.\,'' I take it back! ...\,''

To explain further, one of us (P.D.) has 50 years of trying to get great mathematicians fluent in the language of probability and statistics. See \cite{diaconis2009markov} for one of his efforts. David is the \textbf{only} success for these efforts. It really isn't easy to learn a new language, this same author has spent the past 30 years trying to `learn physics'. Despite small success (publishing in physics journals, teaching in the physics department), no one mistakes him for a physicist. This paper takes one of David's statistical discoveries further back into a basic problem of statistical inference.

\textbf{Our Problem}: A variety of `mixing schemes' are in routine use in applications. To be specific, consider mixing $52$ cards by a `wash' or `smoosh' shuffle. Here the cards are slid around on the table with both hands `for a while', gathered back together and used for a card game. How much mixing is enough? If the mixing isn't good, why not? This mixing scheme is used in EPSN Poker Tournaments, in California Card Rooms and in Monte Carlo for Baccarat. Similar mixing schemes are used with Dominoes and Mahjong Tiles.

Related Problems arise in lottery schemes for generating (say) a random $6$-set out of $\{1,2,\cdots, 75\}$. There, a `blower' randomizes $75$ ping pong balls 'for a while' and six are captured. Before each of the big draws, a limited sample is taken on the day of play (say $200$ $6$-sets) to check that the set-up hasn't been tampered with. Testing random number generators or the output of Monte Carlo algorithms present similar problems.

In general, let $\mathscr{X}$ be a finite set, let $X_1, X_2, \cdots X_N\in \mathscr{X}$ be a sample. We want to test if \{$X_i$\} are drawn from the uniform distribution $u(x)=\frac 1{|\mathscr{X}|}$. In our applications, $|\mathscr X|$ is large, e.g. $|\mathscr{X}|\doteq 8\times 10^{68}$ for cards, $|\mathscr{X}|\doteq 6\times 10^{8}$ for $6$-sets out of $75$. The sample size $N$ is small, e.g. N=200 or at most a few thousand.

It is natural to give up at this point, after all, most probably $X_1, X_2, \cdots X_N$ are distinct and you have $200$ balls dropped into $8\times 10^{68}$ boxes! But, if mostly, the original top card winds up close to the top or most of the $6$- sets contain low numbered balls, it would be poor statistics not to spot this. In other words, to make useful tests, you have to look at the actual practical problem, think about it, and design tests aimed at natural departures from uniformity.

The natural way to incorporate a-priori knowledge is Bayesian statistics. Alas, even in 2017, adapting the Bayesian picture to these problems is a research project. Section two below is a review, designed to bring readers from a variety of backgrounds up to speed. It reviews classical goodness of fit tests, using card shuffling data as illustration. It then reviews Bayesian testing, including `Lindley's Paradox' and the problems with naive, `flat prior' or `objective' Bayesian methods. Finally, needed standard exponential family and conjugate analysis are reviewed, along with the literature on maximum entropy modeling.

Section three presents our main contribution. It uses maximum entropy modeling as in \cite{Zhu:1998:FRF:290091.290092} to build low dimensional models through natural families of test statistics. Conjugate Bayesian analysis of these models is almost fully automated (with a simple `a priori sample size' parameter left to specify). The models include the uniform distribution as a special case ($\theta=0$). A variety of techniques for posterior inference are discussed. These include basic importance sampling to estimate normalizing constants, thermodynamic integration and other bridge sampling schemes. Along with asymptotic approximations (surprisingly useful in our class of examples). The literature here is vast and we attempt a tutorial with useful references.

Section four tries out the various Bayesian procedures on a card shuffling problem and draws some conclusions. As the reader will see, there is much left to do. We only hope that David will join the conversation.

\section{The Basics, A Review}
This section reviews classical topics in statistics. Section \ref{goodnessoffit} treats frequentist goodness of fit tests using shuffling data as illustration. Section \ref{bbayestests} treats Bayes tests, including the problems of classical p-values, Lindley's Paradox and a critique of `flat prior Bayes'. Section \ref{expotool} reviews needed exponential family theory including a description of the Mumford-Wu-Zhu \cite{Zhu:1998:FRF:290091.290092} work on maximum entropy modeling. If you lack statistical background, good references are the standard graduate texts \cite{cox1979theoretical} \cite{ferguson2014mathematical} \cite{lehmann2006testing} \cite{schervish2012theory}.

\subsection{Goodness of Fit Tests}\label{goodnessoffit}

Let $(\mathscr X ,\mathscr B)$ be a measurable space, $\{P_{\theta}\}_{\theta\in \Theta}$ a family of measures indexed by a set $\Theta$. We are given data $X_1, X_2, \cdots, X_N$ assumed to be independent and identically distributed from a fixed probability $Q$ on $\mathscr X$. We ask `Does the model fit the data?' More formally, a test is a function $T(X_1, X_2, \cdots, X_N)\in \{0,1\}$ with $T=1$ interpreted as `doesn't fit' and $T=0$ interpreted as `fits'. A test is consistent if 
\[  \lim_{N\to \infty} P(T_N=1) =1 \qquad \text{if}~~ Q\notin \{P_{\theta}\}_{\theta\in\Theta} \]

As an example, consider the first consistent non-parametric test, the Kolmogorov-Smirnov test. Here $\mathscr X=\R$ and $P_\theta=P$, a single fixed probability. Let $F_N(x)=\frac 1 N \#\{i: X_i\leq x\} $ be the empirical distribution function. Let $\|F_N-P\|=\text{sup}_{-\infty\leq x\leq \infty}|F_N(x)-P (-\infty,x)|.$ A test fixes $\epsilon>0$ and rejects iff $\|F_N-P\|>\epsilon$. Kolmogorov showed that, for any $\epsilon>0$, this test is consistent. Moreover, he gave an elegant way to make a choice of $\epsilon$ by determining good approximations to the distribution of $\|F_N-P\|$ when $Q=P$ and $N$ is large. The distribution of $\|F_N-P\|$ doesn't depend on $P$ and the surprising appearance of theta functions makes for an exciting mathematical story.

Nevertheless, real statisticians know that this is a poor test, it has no known optimality properties (indeed quite the opposite)(\cite{berk1979goodness}). Further, it can take huge sample sizes to reject. For example, D'Agostino and Stephens \cite{d1986goodness} report that if $P$ has the Cauchy density $\frac 1 {\pi}\frac 1{1+x^2}$ and $Q$ is standard normal (density $\frac{e^{-\frac{x^2}2}}{\sqrt{2\pi}}$ ) an order of $300$ samples are needed to detect this evident discrepancy. The standard `good tests' is the Anderson-Darling test. See \cite{d1986goodness} for much more statistical wisdom.

To discuss theory, it is useful to define the \textbf{power} of a test, the probability that $T=1$ given that $Q\notin \{ P_\theta\}$ and the \textbf{ significance level}, the probability that $T=1$ even though  $Q\in \{ P_\theta\}$. One wants tests with high power and small level. Often the  level is fixed at a standard rate e.g. $0.05$ and one seeks a most powerful test at this level. This is the start of statistical theory: Do such most powerful tests exist? Are there many such tests? If so, how can they be distinguished?  

Goodness of fit testing is a well known theoretical statistics nightmare. Indeed the standard `bible', Lehmann's 	`Testing Statistical Hypothesis' didn't include a chapter of goodness of fit testing until its third edition \cite{lehmann2006testing}. Lehmann told us there just wasn't enough theory available. This has changed somewhat due to work of Janson \cite{janssen1999testing} \cite{janssen2000global} \cite{janssen2003power}, summarized in \cite{lehmann2006testing}. But the bottom line is developing goodness of fit tests is much more art than science. The book \cite{d1986goodness} is a fine treatment of the art side.

\begin{exmp}[Smooshing Cards]

Of course, in a practical problem we are allowed to think! Consider mixing by `smooshing'. With  Lauren Banklader and Marc Coram, we gathered some data. A deck of 52 cards was mixed for $t$ seconds ($t=60, 30, 15$). After $t$ seconds, the cards were gathered back together and the current order read off. For each $t$, this was repeated $N$ times. Let us think together, if the cards are not `slid around' for long enough, why wouldn't they be random? We thought that there might be too many cards which were originally together that are still together. This suggested, for a permutation $\sigma$,
\[T(\sigma)=\#\{i, 1\leq i \leq n-1 : \sigma(i+1)=\sigma(i)+1\}.\]

Thus $T(\text{Id})=n-1$,  Some interesting probability theory  \cite{diaconis2014unseparated} shows that, under the uniform distribution, $T(\sigma)$ has an approximate Poisson($1$)  distribution, that is
\[P( T(\sigma)=j)\doteq\frac 1{ej!}\qquad 0\leq j\leq n-1\]
with good control on the error of approximation. From properties of the Poisson(1) distribution, $T(\sigma)$ has mean $1$, variance $1$ and usually $T(\sigma)$ is $0,1,2$. In the present case, $N=52$ (please do not be confused, for a fixed $t$, $52$ distinct permutations of a deck of $52$ cards were generated). This provides the following data (for $t=30$ seconds):
\begin{center}
\begin{tabular}{l*{6}{c}}
$T(\sigma)$            & 0 & 1 & 2 & 3 & 4  & 5  \\
\hline
Observation & 14 & 19 & 12 & 4 & 1 & 2  \\
Expectation & 19.1 &19.1 & 9.6 & 3.2 & 0.8 & 0.2  \\

\end{tabular}
\end{center}
The second row shows $\#\{\text{Data Points with } T(\sigma)=j\}$, the third row shows the Poisson($1$) approximation. The classical $\chi^2$ goodness of fit statistic is
\[
\chi^2=\sum_{j=0}^4(\text{Obs}-\text{Exp})^2/\text{Exp}=6.77.
\]
Here, in accordance with standard practice, categories $[5,\infty]$ are lumped together. Since the data are random, the statistic $\chi^2$ is itself a random variable. Standard theory shows, when $N$ is large, $\chi^2$ has a known limiting distribution, the $\chi_5^2$ ($\chi^2$ distribution with $5$ degrees of freedom). The upper $0.05$ point of this is $11.07$  and a standard test is `reject' if $\chi^2>11.07$, since here $\chi^2=6.77$, the test doesn't reject the null hypothesis. For this example, the `p-value' of the test is $0.248$.
\end{exmp}
In further data analysis of the smooshing data, a variety of test statistics were used. These included position of original top (bottom) card, distance of $\sigma$ from original order in standard metrics\dots For each, the distribution of $T(\sigma)$ under the null hypothesis is derived (exactly or approximately). Comparison of the empirical distribution and theory is made. This raises well studied problems of multiple testing \cite{shaffer1995multiple}. We will return to this example below and conclude with a brief summary. The data for $t=60$ seconds showed no discernible difference from uniformity. The data for $t=30$ seconds was `on the edge' but adjusting for multiple testing, the uniformity assumption is accepted. For $t=15$ seconds, all tests reject uniformity.

All the probability calculations reported above are classical frequentist, thus $p=0.05$ can be interpreted as ``in many independent repetitions, the test will reject the (true) null hypothesis approximately $\frac 1{20}$ times.'' These p-values are under wide spread attack by the statistical community. Spirited discussion of the problems of p-values from a Bayesian perspective are in \cite{sellke2001calibration}, \cite{berger1988statistical} which are recommended reading.

\subsection{Bayes Tests}\label{bbayestests}

To fix ideas, let $\mathscr X$ be a finite set, $\mu$ a probability on $\mathscr X$ and $X_1,X_2\cdots X_N$ independent, identically distributed points from a probability $Q$. One wants to test
\[
H_0: Q=\mu \qquad  \text{vs} \qquad H_1:Q\neq \mu.
\]

Following the classical work of Harold Jeffreys \cite{jeffreys1998theory}, one Bayesian procedure puts a prior probability $P(H_0)$ on the null hypothesis, a probability $P(H_1)$ on $H_1$ and computes the posterior odds $P(H_0|\text{Data})$/$P(H_1|\text{Data})$ and if this is large (e.g. $\geq 20$), the Bayes test accepts the null hypothesis, otherwise the null hypothesis is rejected. Using Bayes theorem
\begin{align*}
\frac{\pnull}{\palt}&=\frac{P (\text{Data}|H_0)}{P (\text{Data}|H_1)}\frac{P (H_0)/P(\text{Data})}{P (H_1)/P(\text{Data})}
=\frac{P (\text{Data}|H_0)}{P (\text{Data}|H_1)}\frac{P (H_0)}{P (H_1)}\\
&=\text{Bayes factor}\times \text{Prior odds}
\end{align*}

Often, for simplicity, the prior odds are taken as $1$ and the Bayes test only depends on the Bayes factor. Of course, one issue is that $P(\cdot)$ must be specified on $H_1$. There are a host of `standard convenience priors' in wide spread use \cite{berger2009formal} discussed further below.

A standard statistical mantra says `with a lot of data, Bayes and frequentist procedures agree'. A classical example of Jeffreys \cite{jeffreys1998theory} shows that there can be discrepancies.
\begin{exmp}[Jeffreys' Paradox]

Consider flipping a $\theta$-coin $n$ times and testing if it is fair $\theta=\frac 12$. The standard frequentist test rejects if $|S_n-\frac n 2|> c\sqrt n$ for appropriate $c$.

Consider a Bayes test with $H_0: \theta=\frac 12$, $H_1: \theta\neq \frac 12$. We could put probability $\frac 12 $ on both $H_0$ and $H_1$. Under the alternative $H_1$, the uniform prior $\text{U}[0,1]$ is put on $\theta$. Then if the data is $S_n=j$, $\bfactor=\frac{{n\choose j}/2^n}{\frac 1 n}$. Using the local Central Limit Theorem, for $j$ close to $\frac n 2$,
\[
\bfactor=\frac n {2^n}{n\choose j}\approx n\frac{e^{-\frac 12 (j-\frac n2)^2/n/4}}{\sqrt{\frac n 4 2\pi}}=c\sqrt n e^{-\frac 12 (j-\frac n2)^2/n/4}
\]

for explicit c.

The test rejects if 
\[
-\frac  {(j-\frac n2)^2}{n/4}+\frac 12 \log n 
\]

is large, e.g. $|j-\frac n 2|>c'\sqrt{n\log n}+c''$

This can make for surprising differences between frequentists and Bayesians as the following example (drawn from Wikipedia) shows.
\end{exmp}

\begin{exmp}[Lindley's Paradox]
In a certain city, $59,581$ boys and $48,870$ girls have been born in a certain time period. The observed proportion of male births is $0.5036$, the frequentist test of $\theta=\frac 12$ based on the normal approximation to the binomial rejects the null hypothesis at the $5\%$ level. The Bayes test has $\pnull=0.95$.
\end{exmp}

\begin{exmp}[The Problem with Flat Priors] Consider our original problem of testing a shuffling procedure for uniformity. Then $\mathscr X$  is all $n!$ permutations and $\mu(\sigma)=\frac 1 {n!}$ is the uniform distribution. The alternative is $H_1=\{\text{All other probabilities on} \mathscr{X}\}$. A standard prior on $H_1$ is the uniform distribution on the simplex of dimension $n!-1$. This is the Dirichlet$(1,1,\cdots,1)$ distribution with $(1,1, \cdots ,1 )$ a vector of $1$'s of length $n!$. The data is $\sigma_1,\sigma_2\cdots \sigma_N$ and in any real situation all of the $\sigma_i$ will be distinct. Then the Bayes factor is 
\[
\frac{\postnull}{\postalt}=(\frac 1 {n!})^N \frac{(n!-1)!}{((n+N)!-1)!}\approx 1
\]
Standard asymptotics shows that e.g. if $n=52$ and $N$ is moderate (e.g. a few thousand), the Bayes factor is $1$ to good approximation. Thus, any data will fail to reject the null hypothesis, so long as there are no repeated values. The same conclusion holds for symmetric Dirichlet priors. The situation is worse for improper priors (e.g. Jeffreys' prior, propotional to $Pod_{\sigma\in S_{n}} x_{\sigma}^{-1}$). For then the posterior remains improper and formally $\postalt =\infty$, so the Bayes factor is zero.
\end{exmp}
The problem is that the uniform distribution on a high dimensional space is not easy to understand (without experience). Consider the uniform distribution on the $n!$ simplex and look at the induced probability of the first coordinate $\sigma(1)$. Here $\sigma(1)$ can take values $1,2\cdots n$ so that we have a probability on the $n$ simplex. Standard properties of the Dirichlet distribution show that
\begin{prop}
Under the uniform distribution on the $n!$ simplex, the induced distribution of $\sigma(1)$ is Dirichlet$_n((n-1)!,(n-1)!,\cdots, (n-1)!)$.
\end{prop}

The induced measure is super-concentrated about $(1/n,1/n,\cdots, 1/n)$ and it would take a huge amount of data to change this. This can be quantified via result in \cite{diaconis1991closed}. Let $Q$ be a random distribution on $\{1,2,\cdots, n\}$ from $D_n((n-1)!,(n-1)!,\cdots, (n-1)!)$. Let $x_{\star}=(1/n,1/n,\cdots, 1/n)$. Then we have $\ee\| Q-x_{\star}\|=\frac 1 {(n-1)!} \frac{\Gamma(n!)}{\Gamma((n-1)(n-1)!)}(\frac{1}{n})^{(n-1)!}(1-\frac 1n)^{(n-1)(n-1)!}$. The same concentration is forced on many other features (e.g. the cycle and descent structures of $\sigma$). We conclude that the uniform doesn't represent anyone's real prior and should be avoided. For several related calculations of strange behavior under the uniform measure on large finite sets, see \cite{diaconis2002bayesian}.

There is some further literature on Bayesian testing, see \cite{verdinelli1998bayesian} and their references. These ideas can be applied to the discrete problems studied here by using an orthonormal basis for $L^2(\mathscr X)$, often available using group representation theory, see \cite{diaconis1989generalization}. Alas, we have not found any suggestions that help with the problems considered here. The methodology put forward in section three use standard machinery of exponential families and conjugate priors. We turn to this next.

\subsection{Exponential Tools}\label{expotool}
Exponential families are families of probability densities proportional to $\{e^{\theta\cdot T(x)}\}_{\theta\in \Theta} $ (more careful definitions in a moment). They capture commonly used statistical models (normal, Poisson, multinomial \dots). Usually, $\Theta\subset \rr^d$ and $T: \mathscr X\rightarrow \rr ^d$ is called a sufficient statistics in the statistics literature. They often turn up in the statistical physics literature in form $e^{-\beta H(x)}$ with $H(x)$ called `energy'. Two basic motivations for their study: First, if $X_1, X_2\cdots X_N$ is a sample from such a family, the data can be `compressed' to $\bar T=\frac 1 N \sum_{i=1}^N T(X_i)$. Up to minor technical assumptions, exponential families are the only families admitting such reductions (the Koopman-Pitman-Darmois Theorem). Second, often $P_\theta(X\in A | T(X)=t)$ is uniform on $\mathcal{X}$, independent of $\theta$. Thus these families can be said to make minimal assumptions beyond having $T$ as sufficient statistic. Otherwise, measuring the entropy of a density as $-\int f(x)\log f(x) dx$, exponential models can be shown to be `maximum entropy with given sufficient statistic'. Or, if $\theta$ is chosen so $\ee_\theta(T(x))=\mu_0$ has a prescribed value. For example, the $N(\mu_0,\sigma^2)$ distribution has the maximum entropy among all measures on $\rr$ with variance $\sigma^2$.

These properties make an attractive package for building models and many statistical schools espouse  and develop these ideas. The work of Martin-L\"of/ Lauritzen \cite{lauritzen2012extremal} and Jaynes \cite{jaynes1957information} has been particularly influential. We were inspired to relook at exponential families because of an innovative study of Mumford, Wu, and Zhou \cite{Zhu:1998:FRF:290091.290092}. They built models for textured patterns in image analysis by taking a host of features of a pattern $T(x)=(T_1(x), T_2(x), \cdots, T_d(x))$ with $T_i$ summary statistics drawn from engineering experience (e.g. If $x$ is an $n\times m$ array of pixels, $T_1(x)$ might be density, $T_2(x)$ the number of horizontally adjacent $(1,1)$ occurrences (including Fourier coefficients). They then fit a model of form $e^{\theta\cdot T(x)}$, estimated $\theta$ and then used samples from the fitted model to generate textured patterns. The results were surprisingly realistic and this effort has been refined and amplified to a powerful modern suite of algorithms. See \cite{lu2015learning}.

When we first saw this work, our reaction was `The nerve of them, this can't possibly work. Just take a large collection of features, stick them up in the exponent without worrying about redundancy of correlations and use standard statistics as if it was a real model!' Their success encouraged us and this kind of brazen modeling is exactly what we propose in the next section.

With these preambles, it is time to get back to mathematics. Any standard graduate text treats exponential families, the books of Barndorff-Nielsen \cite{barndorff2014information}, Brown \cite{brown1986fundamentals}, and Letac \cite{letac1992lectures} are focused texts. The following account is drawn from \cite{diaconis1979conjugate} where more details can be found. Further references to conjugate priors and many examples are in \cite{diaconis2008gibbs}.

Let $(\mathscr X, \mathscr B)$ be a measurable space, $T: \mathscr{X}\rightarrow \rr^d$ a measurable function. Let $\mu(dx)$ be a $\sigma$- finite measure on $(\mathscr X, \mathscr B)$ and let
\[
\Theta=\{\theta\in \rr^d : \int e^{\theta\cdot T(x)}\mu(dx)<\infty\}.
\]
Assume throughout that $\Theta$ is non empty and open. H\"older's inequality implies that $\Theta$ is convex, the natural parameter space. Set $m(\theta)=\log \int e^{\theta\cdot T(x)}\mu(dx)$. The exponential family through $T$ (and $\mu$) is:
\[
P_\theta(dx)=e^{\theta\cdot T(x)-m(\theta)}\mu(dx) \qquad \theta\in \Theta.
\]

Allowable differentiation implies the useful formulas
\begin{align}
& \ee_\theta(T)=\nabla m(\theta)\\
&\text{Cov}_\theta(T)=(\frac{\partial ^2 m(\theta)}{\partial \theta_i\theta_j})_{1\leq i,j\leq d}
\end{align}

\textbf{The Conjugate Prior}. The conjugate prior for $P_\theta(dx)$ has the form
\begin{align}
\pi_{n_0,x_0}(d\theta)= e^{n_0 x_0\cdot \theta-n_0m(\theta)}d\theta
\end{align}
where $d\theta$ is Lebesgue measure. Here $n_0>0$ and $x_0\in \rr^d$ are called the prior sample size and prior mean. In \cite{diaconis1979conjugate} it is shown that if $n_0>0$, $x_0\in \{\text{Interior of  the convex hull of Supp}(\mu^{\intercal^{-1}})\}$. $\pi_{n_0,x_0}$ is a proper probability and $Z(n_0,x_0)$ is a finite normalizing constant. It is further shown that
\begin{align}
\ee_{n_0,x_0}(\nabla m(\theta))=x_0
\end{align}

Thus $x_0$ is the apriori mean of $\ee_\theta(T)$. Bayes theorem shows that the posterior distribution of $\theta$ given $X_1,\cdots , X_N$ is again in the conjugate family with parameters $(n_0+N, \frac{n_0}{n_0+N}x_0+\frac N {n_0+N}\bar{T})$. This linearity is shown to characterize conjugacy. Finally $\pi$ is a log-concave measure centered at its mean, so one has chosen a rough picture of the prior and posterior.

\section{A Practical Bayes Test for Uniformity}\label{sec3}

Let $\mathscr X$ be a finite set, $T_1, T_2,\cdots, T_d$ functions $T_i: \mathscr X\rightarrow \rr$. For $\theta\in \rr^d$, set 
\begin{align}
P_\theta(x)=e^{\sum_{i=1}^d \theta_i T_i(x)-m(\theta)}
\end{align}

a probability density with respect to counting measure. Thus $m(\theta)=\log \sum_x e^{\theta\cdot T(x)}$. Observe that $\theta=0$ gives the uniform distribution on $\mathscr X$. Let $K$ be the interior of the convex hull of $\{T(x)\}_{x\in\mathscr X}$ in $\rr^d$. If the vectors $\{T(x)\}_{x\in \mathscr X}$ are linearly independent, this is a non-empty open set. Otherwise, restrict to a smaller set of $T_i$.

For $n_0>0, x_0\in K$, set 
\begin{align}
\tilde{\pi}_{n_0,x_0}(d\theta)= Z^{-1}e^{n_0 x_0\cdot \theta-n_0m(\theta)}d\theta
\end{align}
a probability density on $\Theta=\rr^d$. To test for uniformity, we center our prior at $x^\star=\ee_0(T)=\nabla m(0)$, this leaves only $n_0$ to be specified. We suggest five possibilities. 
\begin{itemize}
\item One may simply choose $n_0=1$. 

\item Use posterior linearity to make a reasoned choice of $n_0$ along the lines of the device of the imaginary results `if my sample of size $N$ results in $\bar{T}^\star$. I would use $\frac{n_0}{n_0+ N}x_0^\star+\frac N {n_0+N}\bar{T}^\star$'. 

\item A third possibility, widely used (see Good \cite{good1965estimation} for extensive discussion) is to carry $n_0$ along as a parameter and plot the final result as a function of $n_0$. 

\item A fourth approach: $n_0$ can be estimated as usual in empirical Bayes procedures. 

\item The fifth possibility, evaluate or approximate the result as $n_0\to 0$, this is a version of a non-informative prior. The evaluation leads into the rich waters of the geometry of convex polyhedra, see Letac and Massam \cite{letac2012bayes} for remarkable developments in the parallel problem of contingency tables.
\end{itemize}

To compute a Bayes factor, $$\postalt = \frac{Z(n_0+N, \frac{n_0}{n_0+W}x_0^\star+\frac W {n_0+W}\bar{T})}{Z(n_0,x^\star)}$$ must be evaluated. This would usually be done by MCMC samples from the posterior. Since the posterior is conjugate, this is the same problem as sampling from the prior. Another route to drawing conclusions is to simply see if the posterior is supported close to $\theta=0$.

\textbf{Evaluation of Bayes factors:} This is a standard topic in Bayesian statistics. The Wikipedia entry for `Bayes factors' contains a wealth of philosophical as well as theoretical references. The paper by Kass and Raftery \cite{kass1995bayes} is recommended. For our problems, Markov Chain Monte Carlo will often be the default choice. The papers by Meng and Wong \cite{meng1996simulating} and Gelman and Meng \cite{gelman1998simulating} contain useful literature reviews.

One widely used estimator for $\hat P(\text{Data}|H_1)$ , applied below, is the harmonic mean estimator. Draw $\theta_1, \theta_2, \cdots, \theta_M $ from a Markov Chain with stationary distribution the posterior density $\pi^\star(\theta)=P(\theta|\text{Data})$ and use \[
\hat P(\text{Data}|H_1)=\{\frac 1M \sum_{i=1}^M P(\text{Data}|\theta_i)^{-1}\}^{-1}\]

then the Bayes factor can be estimated by 
\[
\hat{\text{BF}}=\frac{ P(\text{Data}|H_0) P(H_0)}{\hat P(\text{Data}|H_1) P(H_1)}
\].

The harmonic mean estimator is easy to calculate and converges almost surely to $ P(\text{Data}|H_1)$. However, it might have high variance and thus not stable, for further discussions, see \cite{kass1995bayes}.

\textbf{Doubly Intractable Priors:} The problem is now manifest: how can we draw from the posterior and evaluate $\hat{\text{BF}}$, when we don't know $m(\theta)$? Of course, there are many ways of sampling from a measure when the normalizing constant is unknown. Here, $m(\theta)$ is a crucial part of the prior and posterior. This has been dubbed the problem of doubly intractable prior distributions \cite{murray2012mcmc}. These authors have proposed an interesting new algorithm, the `exchange algorithm' which we will illustrate in \ref{sec4}. However they involve the ability to draw perfect samples (as per the Propp-Wilson algorithm), we don't see how to do this in our applications for general cases.

In the examples below, we use the straightforward procedure of estimating $m(\theta)$. There are several available routes illustrated below: 
\begin{itemize}
\item Estimate $m(\theta)$ by importance sampling (e.g. for a grid of $\theta$'s) and then interpolate for other values of $\theta$.
\item Estimate $\nabla m(\theta)=\ee_\theta(T(X))$ by sampling from $\pi_\theta$ (e.g. via Metropolis) and then use thermodynamic integration $m(\theta)=\int_0^\theta \nabla m(s) ds +m(0)$, with $m(0)$ available because $P_0$ is uniform.
\item Estimate $Z(\theta)=\sum_t e^{\theta\cdot t}\#\{x: T(x)=t\}$ via exact (or approximate) evaluation of the combinatorial quantity $\#\{x: T(x)=t\}$, we use this in the following section.

\item Of course, there are other approaches, see \cite{kass1995bayes}.
\end{itemize}

\section{A Simple Example}\label{sec4}
This section uses a different card shuffling scheme (random transpositions) to illustrate the Bayesian methodology of Section \ref{sec3}. It begins with motivation and the data then carries out the tests. Following this, a detailed exposition of the exchange algorithm is presented.

The usual method for generating a random permutation on a computer involves picking $I_i$ $(i\leq I_i\leq n)$, and transposing cards $i$ and $I_i$ for $1\leq i\leq n-1$, one of us had a computer programmer who ``made things more random'' by simply making repeated random transpositions. This scheme was also used in the first algorithms for computer generated bridge hands. In each case, the desk size was $52$. The programmer used $100$ random transpositions. The bridge league used $60$ random transpositions. What's the right answer?

A careful mathematical analysis \cite{diaconis1981generating}, \cite{saloff2007convergence} allows the following theorem: let
\[
Q(\sigma)=
\begin{cases}
    \frac 1n,& \text{if } \sigma=\text{Id}\\
    \frac 2{n^2},& \text{if } \sigma=(i,j)\\
    0.  & \text{otherwise} 
\end{cases}
\]
Let $Q^{\star k}(\sigma)=\sum_{\eta}Q(\eta)Q^{\star k-1}(\sigma\eta^{-1})$, let $U(\sigma)=\frac 1 {n!}$

\begin{thm}
For any $n\geq 3$ and $c>0$, let $k= \left \lfloor \frac 12 n\log n+cn \right \rfloor,$ then 
\[
\|Q^{\star k}-U\|\leq 2e^{-c}
\]

with $\|Q^{\star k}-U\|=\text{max}_{A\in S_n} |Q^{\star k}(A)-U(A)|$.
\end{thm}

This shows order $\frac 12 n\log n $ steps are sufficient for a demanding notion of mixing. When $n=52$, $\frac 12 n\log n\doteq 103$. Suppose we want the right side of the bound less than $\frac 1 {1000}$, this needs $k\doteq 480$. There are lower bounds for the theorem too, showing that order $\frac 12 n\log n$ shuffles are needed but these are not close enough to the upper bounds to give useful numerics when $n=52$.

The arguments that follow, which use our approach to Bayesian hypothesis testing show, in a fairly sharp sense, that $k\doteq 180$ is necessary and sufficient. We have $\mathscr X=S_n$, the symmetric group with $n=52$. For fixed $k$, the data are $X_1, X_2, \cdots, X_N$ the results of $k$-random transpositions starting at the identity. In our experiments we took $N=200$ (Later we give results for larger N). Think about the problem de novo: The cards start face down in a row on the table, in order $1,2,\cdots, 52$ left to right. At each step, the left hand touches a card, the right hand touches a card and the cards are transposed. Suppose not enough transpositions have been made. Why wouldn't the resulting permutation be random? A natural answer: There are too many cards left untouched. The resulting permutation has too many fixed points.

In this example, this heuristic can be seen to be very sharply correct. In the original proof, the number of fixed points gave the matching lower bound. Much more refined justifications are available due to work of Schramm, Berestycki and others, see \cite{berestycki2011mixing}, for extensive discussion, of course, in real problems, such heuristics are all we have. Let us follow this up and see where it leads.

Let $F(\sigma)= \# \{i: \sigma(i)=i\}$, the number of fixed points of $\sigma$. A classical theorem (de Montmort) says, under uniformity $F(\sigma)$ has an approximate Poisson(1) distribution:
\[
P(F(\sigma)=j)\doteq \frac 1{ej!} \qquad 0\leq j \leq n.
\]

Following the program of Section \ref{sec3}, set
\[
P_\theta(\sigma)=Z^{-1}(\theta)e^{\theta F(\sigma)}\qquad -\infty<\theta<\infty.
\]

and consider a Bayes test of \[H_0: \theta=0 \qquad \text{vs} \qquad H_1: \theta\neq 0.\] 

For this example, simple combinatorics give $Z(\theta)$:
\[
Z(\theta)=n!\sum_{j=1}^n \frac{(e^\theta-1)^j}{j!}.
\]

But we will not use this below (since it is usually not available). To carry out a Bayes test, a prior for $\theta$ is needed, call it $\pi(\theta)$. Following standard practice, this is centered at $\theta=0$: $\int_{-\infty}^{\infty} \theta \pi(\theta)=0$. In the numerical examples below, we choose a normal prior $N(0,\alpha)$. The conjugate prior can be used in the same way. As usual, the exact choice of the prior doesn't matter in any major way.

Let us begin with the results: Figure \ref{fig:shuffles_comparison} 
shows the histogram of $F(\sigma)$ for 200 randomly generated Poisson(1) random variables, along with the results for 200 data points with $k= 100, 120, 140, 160, 180$ respectively. Figure \ref{fig:frequency_comparison} shows the corresponding frequency comparison between real Poisson and the transposition data. In Figure \ref{fig:shuffles_comparison}, for example, it is obvious from the plot that after $100$ shuffles, only a very small part of the data points have $0$ fix points, which suggests that $100$ shuffles is not enough. However, when the number of shuffles has increased to 160 or 180, then though the histogram and frequency plots (the second row in Figure \ref{fig:shuffles_comparison} and Figure \ref{fig:frequency_comparison}) still do not look exactly the same, we can not confidently reject the null hypothesis (the deck is uniformly shuffled confidently) visually. 

Therefore we carried out the Bayes test, basically we followed the method mentioned in Section \ref{sec3}. For each $k=100, 120, 140, 160, 180$, we take $\pi(\theta)\sim N(0,0.1)$ as the prior and run a $1,000$ step MCMC (using the exchange algorithm, explained below) with a burn-in period of $200$ for $20$ times to sample from the posterior $p(\theta|\text{Data})$, each time we use the harmonic mean estimator (see \cite{kass1995bayes}) to estimate the posterior probability of the null hypothesis $P (H_0|\text{Data})$ and hence estimate the Bayes factor. The following table shows the estimated Bayes factor for different $k$.

\begin{center}
\begin{tabular}{l*{5}{c}}
$k$            & 100 & 120 & 140 & 160 & 180  \\
\hline
Bayes factor & 0 &  0.051 & 5.024 & 18.608 & $+\infty$   \\

\end{tabular}
\end{center}

It turns out that when $k\leq 120$, $P (H_0|\text{Data})\leq 0.05$ (when $k=100$, $P (H_0|\text{Data})$ is around $10^{-19}$ ) and the corresponding Bayes factor is much less than $1$, which means $120$ shuffles is not enough to make the deck of cards uniform. $k=140$ gives $P (H_0|Data)\approx 0.834$ and $k=160$ gives $P (H_0|Data)\approx 0.949$, which means the posterior is slightly in favor of the null-hypothesis. When $k=180$, $P (H_0|Data)$ is essentially $1$, which shows very strong evidence in favor of the null hypothesis.

\begin{figure}[!h]
\centering
\includegraphics [width=0.9\textwidth]{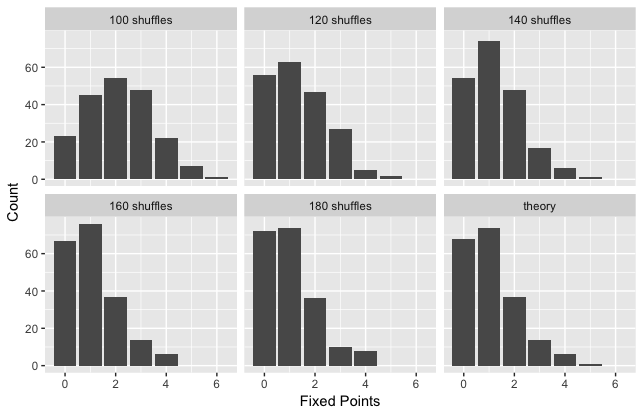}
\caption{\label{fig:shuffles_comparison}Histogram of real Poisson and the transposition data} 
\end{figure}

\begin{figure}[!h]
\centering
\includegraphics[width=0.9\textwidth]{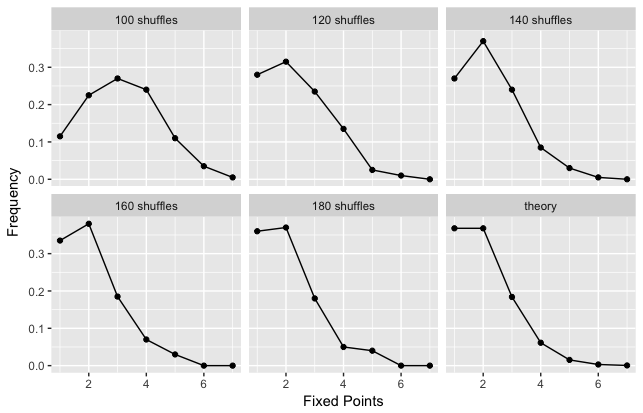}
\caption{\label{fig:frequency_comparison} Frequency comparison between real Poisson and the transposition data} 
\end{figure}

Let us then do the details: since we assumed that the normalizing constant $Z(\theta)$ is unknown, we use Murray's exchange algorithm to sample from the posterior \cite{murray2012mcmc}. In generic form, this begins with a likelihood $p(y|theta)=\frac{f(y|\theta)}{Z(\theta)}$, where $Z(\theta)$ is the unknown normalizing constant, a prior density $p(\theta)$. It runs a Metropolis type Markov Chain on $\Theta$ that begins with a proposal $q(\theta'|\theta)$. From $\theta$:
\begin{itemize}
\item Propose $\theta'\sim q(\theta'|\theta)$
\item Generate the auxiliary $w \sim \frac{f(w|\theta')}{Z(\theta')}$
\item Compute $a=\frac{q(\theta|\theta')p(\theta')f(y|\theta')f(w|\theta)}{q(\theta'|\theta)p(\theta)f(y|\theta)f(w|\theta')}$
\item If $a\leq 1$, go to $\theta'$, if $a<1$ flip an a-coin, if heads go to $\theta'$, else stay at $\theta$.
\end{itemize}

The final factor $\frac{f(w|\theta)}{f(w|\theta')}$ is an unbiased estimator of $\frac{Z(\theta)}{Z(\theta')}$, when $w$ is sampled from $\frac{f(w|\theta')}{Z(\theta')}$, as usually needed for Metropolis. It is not obvious, but this algorithm gives a correct method of sampling from the posterior. The references we looked at did not present a proof of correctness, we include one for correctness.

\begin{prop}
The exchange algorithm above generates a reversible Markov Chain with station distribution the posterior $p(\theta|y)=p(\theta)\frac{f(y|\theta)}{Z(\theta)}$. 
\end{prop}

\begin{proof}
Let $m(\theta,\theta')$ be the probability density of moving from $\theta$ to $\theta'$ in one step. We show \[
p(\theta|y)m(\theta,\theta')=p(\theta'|y) m(\theta', \theta) \]

for all $(\theta,\theta')$.

Writing out the left side, for $\theta\neq \theta'$, the left side equals:
\begin{align*}
&\frac{p(\theta)f(y|\theta)q(\theta'|\theta)}{Z(\theta)}\int \text{min}\{\frac{q(\theta|\theta')p(\theta')f(y|\theta')f(w|\theta)}{q(\theta'|\theta)p(\theta)f(y|\theta)f(w|\theta')}, 1\}\frac{f(w|\theta')}{Z(\theta')} dw\\
&=\frac{1}{Z(\theta)Z(\theta')}\int \text{min}\{q(\theta|\theta')p(\theta')f(y|\theta')f(w|\theta),q(\theta'|\theta)p(\theta)f(y|\theta)f(w|\theta')\}dw
\\
\end{align*}
This expression is symmetric in $\theta, \theta'$, and so equal to $p(\theta'|y)m(\theta',\theta)$.
\end{proof}

\begin{rmk}
Of course, ergodicity of the associated chain must be checked separately. There are a host of variants of Metropolis algorithms in use, see \cite{Billera2001A}. It is not clear if any of these variants work if the trick of replacing $\frac{Z(\theta)}{Z(\theta')}$ by $\frac{f(w|\theta)}{f(w|\theta')}$ with $w$ chosen from $\frac{f(w|\theta')}{Z(\theta')}$ is used. As far as we know, there has been no study of rates of convergence for any instance of these algorithms.
\end{rmk}

A crucial step of the algorithm, is the second step, sampling once from $p(w|\theta')$. The authors of \cite{murray2012mcmc} suggest perfect sampling. If this is available, wonderful. Often it will not be (as in the present case). What we did was write 
\[
P(\sigma|\theta')=Z(\theta')^{-1}e^{\theta'F(\sigma)}
\]
Since the sufficient $F(\sigma)$ is a multinomial distributed random variable which satisfies:
\[
P(F(\sigma)=j|\theta')=e^{\theta'\cdot j}\frac{\#\{\sigma: F(\sigma)=j\}}{52!}\approx e^{\theta' j}/(ej!)
\]
Then we could sample from the multinomial distribution of $F(\sigma)$ given fixed $\theta'$.

For $N=2,000$, a similar exponential family testing is also implemented. Here when k = 160, the Bayes factor $\frac{P(H_0|\text{Data})}{P(H_1|\text{Data})}$ is around $7.285\times 10^{-7}$, when k = 180, the Bayes factor is estimated as $2257.588$, when k = 200, the Bayes factor is estimated as $3.511 \times 10^{37}$. 

To conclude this section, for comparison we carry out a more classical Bayesian analysis for the same choices of $k$ when $N=2,000$.  We use a conjugate prior $\Gamma(\alpha,\alpha)$ for Poisson distribution to test if the fix points data satisfies Poisson(1). Here we take the prior $\Gamma(\alpha,\alpha)$ to enforce the prior has mean $1$, since we are testing if the Possion$(\lambda)$ distribution has parameter $\lambda=1$. Figure \ref{fig:BF_160},\ref{fig:BF_180}, and \ref{fig:BF_200} are the corresponding Bayes factors as a function of $alpha$. For $k=160, 180, 200$, a sharp transition occurs, for $k=160$, the Bayes factors are tiny which suggests $160$ shuffles is not enough. For k=180, the Bayes factor stay bounded and makes mixing plausible. For k=200, Bayes factors are large.

\begin{figure}[!h]
\centering
\includegraphics[width=0.8\textwidth]{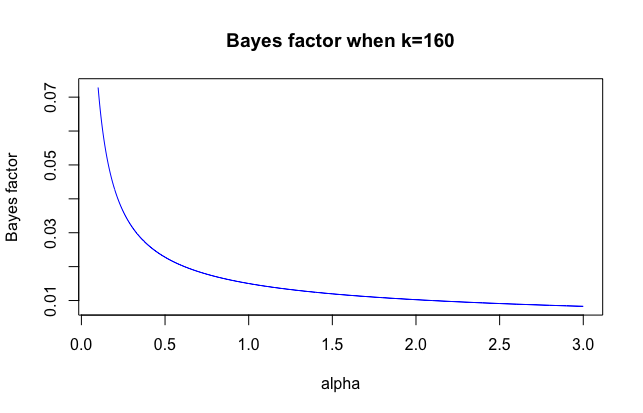}
\caption{\label{fig:BF_160}Bayes factor as a function of $\alpha$ when k=160} 
\end{figure}

\begin{figure}[!h]
\centering
\includegraphics[width=0.8\textwidth]{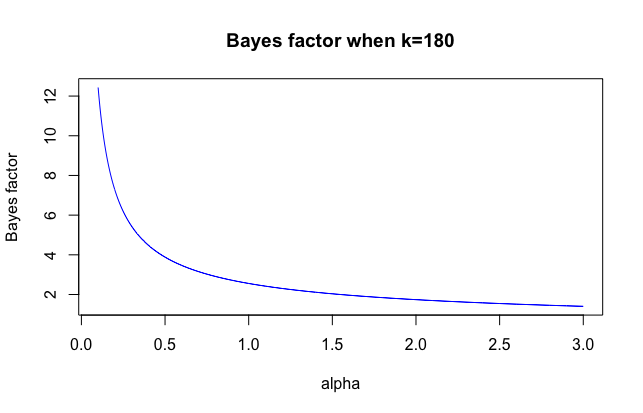}
\caption{\label{fig:BF_180}Bayes factor as a function of $\alpha$ when k=180} 
\end{figure}

\begin{figure}[!h]
\centering
\includegraphics[width=0.8\textwidth]{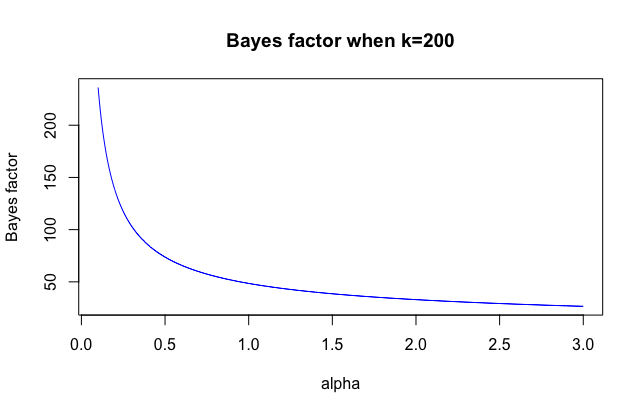}
\caption{\label{fig:BF_200}Bayes factor as a function of $\alpha$ when k=200} 
\end{figure}

\textbf{Main Contribution} The main contribution of this paper is calling attention to the many examples where a similar program can be carried out. Even if $\mathscr X$ is huge, if $T(X)$ takes a limited number of values, and if an analytic approximation is available for $\#\{x: T(x)=t\}$, then the program can go forward.

Here is a second example: for the wash shuffle example, say $T=(T_1,T_2,T_3)$ with $T_1(\sigma)=\#\{i: \sigma(i+1)=\sigma(i)+1\}$ (number of adjacent values), $T_2(\sigma)=\sigma^{-1}(1)$ (position of original top card), $T_3(\sigma)=\sigma^{-1}(n)$ (position of original bottom card). Under uniformity, $T_1\approx$ Poisson(1), $T_2\approx $ Uniform, $T_3\approx $ Uniform. With $T_1, T_2, T_3$ approximately independent, we plan to carry out the analysis in later work.

\section*{Acknowledgements}
The authors thank Paulo Orenstein, Sergio Bacallado, Susan Holmes and Joe Blitzstein for many helpful comments.
\bibliographystyle{alpha}
\bibliography{di}

\end{document}